\documentclass[10pt,journal]{IEEEtran}

\usepackage{graphicx}
\usepackage{epstopdf, epsfig}
\usepackage{amsmath}
\usepackage{amssymb,amsthm,amsfonts}
\usepackage{enumerate}
\usepackage{enumitem}
\usepackage{setspace}
\usepackage{tcolorbox}
\usepackage{cite}
\usepackage{acronym}
\usepackage{xcolor}

\graphicspath{{figures/}}

\newcommand{\HighLight}{\textcolor{black}}

\usepackage{fancyhdr} 
\begin{document}
\title{Maximizing Service Reward for Queues with Deadlines}

\author{Li-on Raviv, Amir Leshem
\thanks{Li-on Raviv is with the Faculty of Engineering, Bar-Ilan University, Ramat Gan 52900, Israel. Email: li.raviv@hotmail.com}
\thanks{Amir Leshem is with the Faculty of Engineering, Bar-Ilan University, Ramat Gan 52900, Israel. Email: leshem.amir2@gmail.com}
\thanks{This paper was partially supported by the Israeli Innovations Authority as part of the HERON 5G Magnet Consortium and the ISF-NRF reserch grant no. 2277/2016.}
}

\date{}  
\newcounter{AssumptionsCounter}


\newtheorem{definition}{Definition}
\newtheorem{lemma}{Lemma}
\newtheorem{theorem}{Theorem}
\newtheorem{claim}{Claim}
\newtheorem{remark}{Remark}
\newtheorem{observation}{Observation}
\maketitle

\begin{abstract}
In this paper we consider a real time queuing system with rewards and deadlines. We assume that packet processing time is known upon arrival, as is the case in communication networks. This assumption allows us to demonstrate that the well known Earliest-Deadline-First policy performance can be improved. We then propose a scheduling policy that provides excellent results for packets with rewards and deadlines. We prove that the policy is optimal under deterministic service time and binomial reward distribution. In the more general case we prove that the policy processes the maximal number of packets while collecting rewards higher than the expected reward. We present simulation results that show its high performance in more generic cases compared to the most commonly used scheduling policies.
\label{sec:abstract}
\end{abstract}


\section{Introduction}
\label{sec:introduction}
The growing usage of the Internet as an infrastructure for cloud services and the Internet of Things (IoT) has resulted in an extremely high demand for traffic bandwidth. The growth of the infrastructure has not kept pace with demand. This gap creates queues in the downstream traffic flow, which impacts service quality in terms of latency and packet loss. The Internet provides service for various types of applications that differ in their priority and latency sensitivity. In this paper, we discuss ways to schedule packets based on their sensitivity to latency (deadline) and priority (reward). We present an a novel scheduling policy that uses the information about the packets upon arrival. We show that this scheduling policy outperforms other policies which process maximal number of jobs.

Barrer in \cite{barrer1957queuing1} and \cite{barrer1957queuing2} pointed out that customers in a queue often have practical restrictions; namely, the time they are willing to wait until they get service.
This problem is known as the impatient customer phenomenon when it refers to people. 

A similar problem exists when handling time deadlines related to job processing or packet data transmission. The problem is known as real time system \cite{lehoczky1996real, stankovic2012deadline}. The term real-time system is divided into soft real time and hard real time. In the soft real time there is a non-completion penalty while  in hard real time there is no reward on non-completion. If the penalty is higher than the reward a soft real time system behaves similar to hard real time system scheduling \cite{yu2016deadline}. Brucker in \cite{brucker2007scheduling} describes many models of queuing systems and cost functions which use weights and due dates. While the weights are rewards and the due dates are deadlines as a step function the environment is limited integer values only. Basic overview of queuing theory to delay analysis in data networks is described e.g. in Bertsekas and Gallager \cite{bertsekas1992data} and in \cite{srikant2013communication}. These delay models focus mainly on the changing environment of the arrival process, service time, and the number of servers.

In the communication world, many Internet applications vary in their sensitivity to the late arrival of packets. The application are divided into two types; services with guaranteed delivery (TCP based applications) and non guaranteed delivery (UDP, RTP, ICMP based  applications).  In both cases there is an expiration time. VoIP, video, and other real-time signaling are highly sensitive to packet delays. These RTP based applications transfer data without acknowledgment (non-guaranteed delivery). In IoT environment sensors information needs to arrive to the processing units on time in order to ensure correct decisions. In these cases the timeout is short and late packets are ignored.  
Guaranteed delivery applications like File Transfer Protocol (FTP) do not suffer as much from packet delay since their deadline is relatively long. Packets that arrive after the long deadline are considered to be lost, and lost packets are considered to be late \cite{katzirscheduling}. In addition, in guaranteed based protocols, the round trip of payload delivery and acknowledgment is important to avoid re-transmission. The resulting re-transmissions overload the network. Detailed analysis of the Internet applications and latency problem is described in \cite{briscoe2016reducing}. Therefore, to maintain high-efficiency networks, late packets must be dropped as soon as it is known that they will miss the deadline.

There are many models to the impatient customer phenomenon in the literature. One common approach is to view impatient customers as an additional restriction on customer behavior, as was presented in Stankovic et al \cite{stankovic1995implications}. Human behavior assumes no knowledge  about the patience of a specific customer. When assessing balking behavior, it is assumed that a customer decides whether to enter the queue by estimating waiting time (balking). In \cite{wang2010queueing} the customer leaves the queue after waiting (reneging). There are models in which the actual deadline is known upon arrival \cite{bhattacharya1989optimal, peha1996cost, towsley1991optimality, panwar1988optimal}.  Other models assume only statistics regarding the abandonment rate \cite{wang2010queueing, larranaga2015asymptotically, atar2010cmu, ayesta2011nearly, yu2016deadline}. In some cases the abandonment included also cases of server failure which forces synchronized abandonment as described in \cite{altman2006analysis, adan2009synchronized, kapodistria2011m}.
The analysis of impatient customer behavior can be measured in two ways. One way measures the Impatience upon End of Service (IES) as was presented in \cite{movaghar1996queueing}.  In this approach a client must end its service before its deadline expires in order to be considered “on time”. The second measures the Impatience to the Beginning of Service (IBS) as is discussed in  \cite{movaghar2006queueing}.  If the service time is deterministic or exponential, the results of IES are also applicable to the IBS system without changing the expiration distribution behavior. Another type of prior information provided to the scheduler is the service time. As in the deadline case there might be statistical knowledge about the service time \cite{ larranaga2015asymptotically, atar2010cmu, ayesta2011nearly, yu2016deadline} in other cases there is no information about the service time \cite{ panwar1988optimal}.  An explicit information about a service times may be provided either upon arrival, prior to receiving service or even not known at the beginning of the service. In the case of goods and data packets, this information is already known at the time of arrival in the queue, and can be beneficial. 

The EDF scheduling is the most popular policy for scheduling real-time systems \cite{stankovic2012deadline}. The policy was proved to be optimal in many scenarios using different metrics. It was shown that if a queue is feasible then EDF is optimal by Jackson's rule \cite{stankovic2012deadline}, meaning that in the level of queue ordering it is optimal. In the general case of non feasible queues Panwar, Towsley and Wolf \cite{panwar1988optimal} showed that if unforced idle times are not allowed, EDF is optimal for the non-preemptive M/D/1 queue when the service time is exactly one time unit, EDF is shown to be optimal assuming service times are not known at the beginning of the service.
In \cite{towsley1991optimality} they expanded the proof to multiple queues when only partial information was provided on the deadlines. In \cite{katzirscheduling} Cohen and Katzir show that EDF without preemption is an optimal scheduling algorithm in the sense of packet loss in the VoIP environment assuming fixed packet size. Other applications of the EDF were discussed in \cite{saleh2010comparing}, comparing the EDF and FCFS (First Come First Serve) in networks with both QoS and non-QoS requirements.
In that paper different rewards are added to a successful packet delivery. This mechanism is used to differentiate between types of packets, and allows the scheduler to assign different service levels. The reward mechanism has a financial side which can involve setting Internet Service Provider (ISP) fees for transferring Internet traffic according to different Service Level Agreements (SLAs).
In the late 1990's, two standards were implemented to support QoS over the Internet. The IETF proposed two main architectures, IntServ \cite{braden1994integrated} and DifServ  \cite{blake1998architecture} working groups. 

A related problem is the multi-class scheduling which assumes mainly statistical information about the jobs. The problem consists of a single or multiple servers that can process different classes of customers arranged in different queues (each class has a queue). The queues are ordered according to FCFS. It is assumed that there is a cost ($c_k$) associated with each queue that is proportional to the waiting time of the customers and a service rate $\mu_k$. The target is to minimize the cost. It is assumed that the service time is unknown to the scheduler.  Atar, Giat and Shimkin expanded the model to include abandonment \cite{atar2010cmu}. They introduced the $c\mu/\theta$-rule, which gives service to the class k with highest $c_k\mu_k/\theta_k$ where $\theta_k$ is the abandonment rate. Customer patience is modeled by an exponentially distributed random variable with mean $\frac{1}{\theta_k}$.
The $C\mu/\theta$-rule minimizes asymptotically (in the fluid limit sense) the time-average holding cost in the overload case. In \cite{ayesta2011nearly} it was shown that the problem can be formulated as a discrete time Markov Decision Process (MDP). Then, the framework of multi-arm restless bandit problem is used to solve the problem. An optimal index rule for 1 or 2 users was obtained for users with abandonment. The queue selection problem is solved in \cite{kim2013dynamic} by approximating it as Brownian control problem under the assumption that the abandonment distribution for each customer class has an increasing failure rate.

Scheduling policies such as FCFS, Static Priority (SP), and EDF focus on the scheduling mechanism, which only defines the job service order. A different approach is to define which jobs do not receive service; these policies are known as dropping policies. The best results are achieved by combining scheduling and dropping approaches. The Shortest Time to Extinction (STE) policy is an algorithm that orders the queue while implementing a dropping mechanism for jobs that definitely cannot meet the deadline. The studies reported \cite{panwar1988optimal} and \cite{bhattacharya1989optimal} discuss the idea of implementing both $STE$ scheduling and a simple dropping mechanism which drops jobs that are eligible for service if their deadline expires before receiving service. In this paper we refer to EDF as $STE$.
Clare and Sastry \cite{clare1989value} presented a value-based approach and reviewed several policies and heuristics that use both scheduling and dropping mechanisms.

In order to evaluate the performance of these policies, a quantitative measure or a metric that reflects network performance needs to be defined. There are two approaches to defining a quantitative measure. The first approach is based on measuring the rewards by a utility function whereas the second approach is based on the cost function. The reward function measures the sum of the rewards of the jobs that received service. The cost function measures the sum of the rewards of the jobs that did not receive service. Peha and Tobagi \cite{peha1996cost} introduced cost-based-scheduling (CBS) and cost based dropping (CBD) algorithms and their performance. 

In general traffic models are divided into light traffic, moderate traffic and heavy traffic. This classification is derived from the ratio of the arrival process and the required service time. 
Examples for heavy traffic analysis can be found in \cite{kingman1962queues}, \cite{kingman2003single} and \cite{atar2017workload}.  In general heavy traffic assumes $\rho > 1$ which forces the queue to explode. In this paper the queue size is stable even if $\rho > 1$ since it is bounded. 

In this paper we propose a novel queuing policy that provides excellent results and prove it is optimal in certain cases. 

The queuing model that is presented in the paper assumes:
\begin{itemize}
\item Jobs parameters are: arrival and service times, deadlines and rewards.
\item All realizations are known upon arrival. This approach is less suitable for human behavior, but reflects the behavior of goods for delivery and packets over networks.
\item The server is idle infinitely often. This assumption is used in many proofs that focus on the interval between server idle times.
\item The scheduling policy is non-preemptive and forcing idle time is not allowed
\end{itemize}

We prove that this policy accumulate more rewards than previously proposed policies such as EDF and CBS, under deterministic service time. We proved that in the case of dual reward and deterministic time the policy is optimal. The same proof also shows that the policy is optimal in the case of dual service times and constant reward.
The proposed approach significantly improves the performance of routers due to its low computational demands.

The structure of the paper is as follows: In section \ref{sec:model} we present the queuing model, its assumptions and its notation. In section \ref{sec:policies} we describe the EDF policy and $c\mu/\theta$ policy. We also provide a simple example of how modifying the EDF policy improves the performance as the processing time is provided upon arrival. In section \ref{sec:MUD} we introduce the Maximum Utility with Dropping scheduling policy and analyze its performance. In section \ref{sec:perfanal} we analyze the performance of the proposed policy. 

In section \ref{sec:num} we present the numerical results and conclude the paper in section \ref{sec:conclusion}.

\section{Queuing with rewards and deadlines}
\label{sec:model}
In this section we describe the mathematical model of queues with deadlines and rewards. We setup notation and describe the various assumptions required for our analysis.

A queuing system is composed of jobs, scheduling policy queues and servers. The jobs entering the system are described by a renewal process with different attributes. The attributes may be assigned upon arrival or later on. The scheduling policy is responsible for allocating the servers to the jobs and choosing the job to be processed out of the queue when the server is idle. In the next few sections we define the stochastic model and the system events.

In this system events occur in continuous time following a discrete stochastic process. The event timing allows for a discrete time presentation. The renewal process is defined as follows:
\begin{itemize}
\item Let $\mathcal{A} =\{ \mathcal{A}_0, \mathcal{A}_1, \mathcal{A}_2..\mathcal{A}_i..\} $ be a timely ordered set of job arrival events.
\item Let $t_i$ be the time of $\mathcal{A}_i$.
\item Let $\mathcal{A}_0$ be the initialization event, $t_0=0$.
\item Let $\mathcal{A}_1$ be the first arrival to the system.
\item $\forall i \in  \mathbb{N}:t_i < t_{i+1}$.
\item Let $J_i$ be the job that arrived at time $t_i$.
 \end{itemize}

We use the extended Kendall \cite{kendall1953stochastic} notation $A/B/C-D/W$ to characterize the input process and the service mechanism.

\HighLight{
The random processes has different distribution types as $GI^+$ for a General i.i.d distribution with strictly positive values, $M$ for a Markovian process and $D$ for a deterministic distribution. $B$ is used for binomial i.i.d distribution of two different values, (the first value occurs in probability $p$ and second value occurs in probability $1-p$).}

We define the attributes of job $J_i$ as follows:
\begin{itemize}
	\item Let $A_i$ be the inter arrival time of the renewal-reward process $\mathcal{A}$, $ t_i = t_{i-1} + A_i= \sum\limits_{j=1}^i A_j$.
	\item Let $B_{min} \leq B_i \leq B_{max}$ be the processing time required by $J_i$.
	\item Let $D_{min} \leq D_i \leq D_{max}$ be the deadline of $J_i$. The time is measured from the arrival time to the beginning of the service (IBS).
	\item Let $W_{min} \leq W_i \leq W_{max}$ be the reward of processing $J_i$.
\end{itemize}

If deadline is large enough it becomes redundant.  In stable queues the bound can be the maximal queue size times the average service time according to Little's law. 
We use the following assumptions as part of the queuing model:
\begin{enumerate}[label=A\arabic*]
 	\item \label{A1} $B_i$, $D_i$ and  $W_i$ are known upon arrival of $J_i$.
\setcounter{AssumptionsCounter}{\value{enumi}}
\end{enumerate}
\HighLight{
\ref{A1} is fundamental assumption in this paper and the proposed policy operation depends on it. In many application assumption \ref{A1} is reasonable, for example in routers, where the service rate is known as well as packet size. To analyze the performance of the proposed method we need further assumptions: \ref{A2}-\ref{A8}. However these are only required to simplify the analysis.}

When $B_i$ is deterministic we use $B_i = 1$ as the service time.

The server is either in busy state when it is servicing a customer, processing a job or transmitting a packet. If the server is not busy it is in the idle state. In this article we use the job processing terminology. We assume that the server is always available (not in fault position).

The service mechanism and the scheduling policy follow the assumptions below:
\begin{enumerate}[label=SA\arabic*]
	\item \label{B1} A single server in the model  $(C=1)$.
	\item \label{B1A} The server processes a single job at a time.
	\item \label{B6} Non-preemptive policy - Processing must not be interrupted.
	\item \label{B3} Forced idle times are not allowed - The server processes jobs as long as there are eligible jobs in the queue without delays.
	\item \label{B2} The policy chooses the next job to be processed from the queue.
	\item \label{B5} The policy is deterministic. For the same jobs in the queue, the service order is always the same.
\end{enumerate}
This article focuses on a single server problem as discussed in \cite{hopp2008single} with a renewal reward process and a general distributed service time models with boundedness assumptions.

The events that occur in the model are:
\begin{enumerate}[label=EV\arabic*]
	\item \label{E1} Job arrival  ($\mathcal{A}_i = \{\ref{E1}\}$).
	\item \label{E2} Job processing begun.
	\item \label{E3} Job processing completed.
	\item \label{E4} Job dropping - a policy decision not to process a job forever. At this point we assume it leaves the queue.
	\item \label{E5} Job deadline expired - In this event the policy must drop the job.
\end{enumerate}

Job arrival (\ref{E1}) and deadline expiration (\ref{E5}) are events that are independent of policy decisions. Beginning the processing of a job (\ref{E2}) and job dropping (\ref{E4}) depend on policy decisions. Job processing completion (\ref{E3}) is derived from the initial processing time and the total processing time (which, by our assumptions, is known before the beginning of the processing). A job deadline expiration (\ref{E5}) has no immediate effect on the system. If a job deadline has expired (\ref{E5}) the policy must drop the job at the some point (\ref{E4}). In our model we allow job drop (\ref{E4}) to be only upon arrival (\ref{E1})  or at job processing begun (\ref{E2}) events. Policy can decide to drop a job before its deadline has expired. Assuming forcing idle time is forbidden (\ref{B3}), job processing completed takes place either when the queue is empty or at job processing begun (\ref{E2}). The events in the system are defined as follows:

\begin{itemize}
\item Let $\mathcal{E}^\pi =\{ \mathcal{E}_0, \mathcal{E}_1, \mathcal{E}_2^\pi..\mathcal{E}_n^\pi..\}$ be a timely  ordered set of all events that occur according to the renewal reward process $\mathcal{A}$ and policy $\pi$.
\item $\mathcal{E}_0=\mathcal{A}_0$ and $\mathcal{E}_1=\mathcal{A}_1 \cup \{\ref{E2}\}$.
\item $\forall n > 1: \emptyset \neq \mathcal{E}_n^\pi \subseteq \{\ref{E1}, \ref{E2},  \ref{E3}\}$.
\item Let $\hat{t}_n$ be the time of $\mathcal{E}_n^\pi$.
\item $\forall n \in  \mathbb{N}:\hat{t}_n < \hat{t}_{n+1}$.
\item $\forall i > 1, \exists n>1$ such that $t_i = \hat{t}_n$ and $\mathcal{A}_i \subseteq \mathcal{E}_n^\pi$.
\end{itemize}

\begin{definition}
The following parameters enable the analysis of the service at time $t$.
\begin{itemize}
\item Let $\hat{Q}_{t}^{\pi}$ be the set of waiting jobs in the queue at time $t$.
\item Let $Q_{t}^{\pi} \subset \hat{Q}_{t}^\pi$ be the queue potential at time $t$. Queue potential is an ordered set of jobs according to $\pi$ that will get service if no new job goes into the queue. By definition the queue potential is feasible.
\item Let $\hat{S}_{t}^{\pi}$ be the set of processed jobs by policy $\pi$ up to time $t$.
\item Let $S_{t}^{\pi} = Q_{t}^{\pi} \cup \hat{S}_{t}^\pi$ be the set of potential processed jobs.
\item Let $\Delta Q_{t}^{\pi} = |Q_{t}^{\pi}| - |Q_{t'}^{\pi}| ,(\Delta Q_{t}^{\pi} \in \mathbb{Z})$ be the difference in number of jobs in the queue potential between two consecutive system events where $t=\hat{t}_n$ and $t'=\hat{t}_{n-1}$.
\end{itemize}
\end{definition}
If $D_{max}$ and $B_{min}$ are finite the queue length is bounded by:
\begin{equation} \label{Qbound}
|Q_{t}^\pi| \leq \frac{D_{max}}{B_{min}}
\end{equation}

The General Queue Events of the system are:
\begin{enumerate}[label=GQ\arabic*]
\item A job reaches to the system ($\mathcal{A}_i \subseteq \mathcal{E}_n^\pi$).
\begin{enumerate}
\item If the server is idle and the queue is empty  $(\{ \ref{E1}, \ref{E2}\} \subseteq \mathcal{E}_n^\pi)$ it gets service immediately and enters the set of processed jobs ($\hat{S}_{t_i}^{\pi}$).
\item If the server is busy ($\ref{E1} \in \mathcal{E}_n^\pi, \ref{E2}\ \notin \mathcal{E}_n^\pi$), it enters the queue. At this point it is possible to compute the queue potential ($Q_{t_i}^{\pi}$). If the new job is added to the queue potential,  other jobs may be removed from the queue potential.
\end{enumerate}
\item The server has completed processing a job ($\mathcal{A}_i \nsubseteq \mathcal{E}_n^\pi$).
\begin{enumerate}
\item If the queue is not empty ($\mathcal{E}_n^\pi=\{\ref{E2}, \ref{E3}\}$) then the job that is at the head of the queue is processed and moves from the queue and queue potential to the set of processed jobs.
\item If the queue is empty ($\mathcal{E}_n^\pi=\{\ref{E3}\}$), there is no change in the set of system parameters.
\end{enumerate}
\end{enumerate}

\begin{figure}
  \centering
  \includegraphics[width=2.5in]{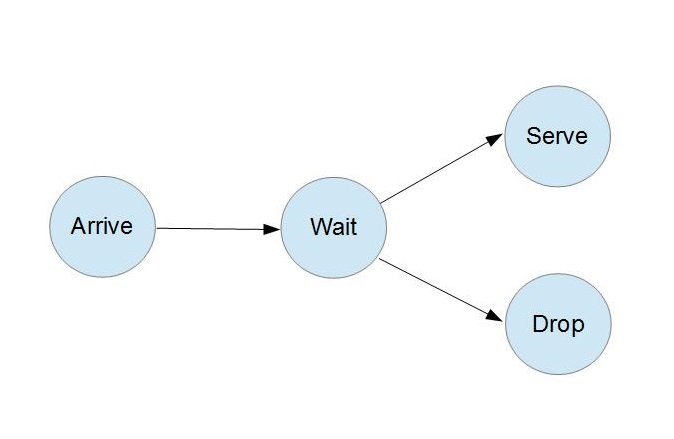}
  \caption[center]{System job processing flow}
\label{Flow2}

\end{figure}

We define a job as a tuple $J_i = < a_i, b_i, d_i, w_i, e_i>$.

\begin{itemize}
	\item $a_i, b_i, d_i$ and $w_i$ are realizations of $A_i, B_i, D_i$ and $W_i$.
	\item $e_i = t_i + d_i$
\end{itemize}

\begin{definition}\label{D1A}
The cumulative reward function for time $t$ and policy $\pi$ is:
\begin{equation}
U_{t}^\pi = \sum\limits_{J_i \in S_{t}^{\pi}}^{}w_i
\end{equation}

Let $t=\hat{t}_n$ and $t'=\hat{t}_{n-1}$ then, the reward difference function is:
 \begin{equation}
\Delta U_{t}^{\pi} = U_{t}^\pi - U_{t'}^\pi
\end{equation}
\end{definition}

The objective is to find a policy $\pi$ that maximizes the cumulative reward function.
If the rewards are deterministic and $W_i = 1$ the cumulative reward function measures the number of jobs that got service. If the rewards are $W_i=B_i$ the cumulative reward function measures the processing time utilization.

In Appendix \ref{AnnexA} queue behavior analysis is provided. As a result of this analysis, it is enough to examine the size of the queue potential and deduce the relationship between the corresponding size of $\hat{S}_{\hat{t}_n}^\pi$. Note that jobs that currently do not meet their deadline are disregarded here since they are not part of the queue potential.

\section{Scheduling Policies}
\label{sec:policies}
In this section we describe several relevant scheduling policies which treat queues with priorities and deadlines.
\subsection{Earliest Deadline First}
We define the EDF as follows: Let $t$ be the current time and assume that the queue is not empty and the server is ready to process a job.
\begin{tcolorbox}
\center{Earliest Deadline First policy}
\begin{enumerate}
	\item $J_i := \underset{J_j \in Q_{t}^\pi} {argmin} (e_j)$ .
	\item If $e_i < t$ then drop $J_i$
	\item else provide service to $J_i$
	\item Return to state 1
\end{enumerate}
\end{tcolorbox}

The EDF was proved to be optimal under different metrics  \cite{panwar1988optimal, bhattacharya1989optimal, katzirscheduling, stankovic2012deadline}. Because of this optimality the EDF became the standard policy used in queuing models with deadlines.

\HighLight{
Panwar, Towsley, and Wolf distinguish between service times being assigned  (A-i) at arrivals or (A-ii) at beginning of services.  Independently of when the assignment occurs, they assume that (A-iii) service times are not known to the scheduler at the beginning of the service. Then, they show that, under (A-iii), the two assignments (A-i) and (A-ii) are equivalent.  In this paper, in contrast, we consider the case where (A-iii) does not hold (see assumption \ref{A1}).
Then, the distinction between (A-i) and (A-ii) is important. To illustrate that knowing the service times upon arrivals can lead can lead to improved results compared to EDF we consider the following variation on EDF which we term MEDF.}

Assume that there are two jobs in the queue. $J_1$ has its expiration time at time 10 seconds and a service time of 30 seconds. $J_2$ has its expiration time at 20 seconds and a service time of 5 seconds. According to the EDF policy $J_1$ is processed before $J_2$. In this case $J_2$ is expired and is lost. If we change the order of $J_1$ and $J_2$ can be executed. Figure \ref{MEDF1} depicts this case.

\begin{figure}[h!]
  \centering
  \includegraphics[width=2.5in]{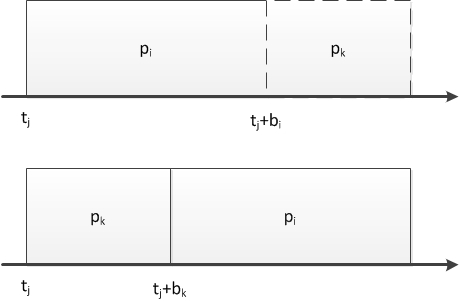}
  \caption[center]{Example of queue order}
\label{MEDF1}
\end{figure}

The MEDF policy is a simple policy based on EDF but uses knowledge of the job reward and service time. It replaces the order of the first two jobs in the queue if this is beneficial. We introduce it here in order to present the advantages of using  information on the service time.

Assume the current time is $t$. The queue has at least two jobs waiting and just finished processing a job and needs to choose a new job to process.

\begin{tcolorbox}
\center{MEDF}
\begin{enumerate}
	\item $J_i := \underset{J_k \in \hat{Q}_j^{MEDF}} {argmin} (e_k)$
	\item $J_l := \underset{J_k \in \hat{Q}_j^{MEDF}-\{J_i\}} {argmin} (e_k)$ 
	\item If $(e_i < t)$ then drop $p_i$
	\item else if $(e_i \geq (t + b_l)) \wedge (e_l < (t + b_i))  $ then provide service to $J_l$ 
	\item else provide service to $J_i$
	\item Return to state 1
\end{enumerate}
\end{tcolorbox}

\subsection{$c\mu/\theta$  policy}
Another policy which is relevant as a benchmark  for testing our algorithm is the $c\mu/\theta$ scheduling policy \cite{atar2010cmu}. The $c\mu/\theta$ policy assumes that are $K$ queues in the system where it needs to select the queue to be served. We assume that there are $K$ levels of rewards (or $K$ group of rewards) each against a queue. The $c\mu/\theta$ scheduling algorithm is composed of two steps:
\begin{enumerate}
\item Upon arrival the policy inserts the new job to queue $k$ according to its reward in FCFS order. 
\item If the server becomes idle the policy chooses the queue to be served according to $c_k\mu_k/\theta_k$ and process the job that is in the head of the queue.
\end{enumerate}

\begin{tcolorbox}
\center{$c\mu/\theta$ arrival}

{Assume new arrival of job $J_i$.}
\begin{enumerate}
	\item $k := w_i$
	\item Insert the $J_i$ at the end of queue $k$.
\end{enumerate}

\center{$c\mu/\theta$ service}
\begin{enumerate}
       \item $k_0 := \underset{k=1..K} {argmax} (\sum c_k\mu_k/\theta_k )$ 
	\item $i := min \{ l : J_l \in \hat{Q}_t^{k_0}\}$	
	\item If $(e_i < t)$ then drop $J_i$
	\item else provide service to $J_i$
	\item Return to state 1
\end{enumerate}
\end{tcolorbox}

The $c\mu-rule$ was proven  to be optimal in \cite{buyukkoc1985c} assuming geometric service times and a preemptive scheduler. In \cite{bispo2013single} Bispo claimed that for convex costs, the optimal $c\mu-rule$ policy depends on the individual loads. He proposed first-order differences of the single stage cost function that reaches near optimal performances. In the non-preemptive case, the $c\mu-rule$ was proved to be optimal  for the two-user case in \cite{kebarighotbi2009revisiting}. The proof was done using calculus of variations arguments as well as an Infinitesimal  Perturbation Analysis  (IPA) approach.  Down, Koole and Lewis \cite{down2011dynamic} provided an optimal server assignment policy, under reward, linear holding costs, and a class-dependent penalty which reduces the reward for each customer that abandons the queue.

In our model we assume that the service times, deadlines and rewards are know upon arrival. The EDF assumes that there is no information about the service times, while in $c\mu/\theta$ the deadline is known statistically in terms of the probability of abandonment. This information has advantages in the case of a non-deterministic service time.  We present a simple example showing the advantages of using knowledge of the service time upon arrival in the EDF case and present new version of $c\mu/\theta$ that exploits the knowledge of deadlines information.
In the case of the $c\mu/\theta$ policy, knowing deadlines upon arrival allows us to modify the queue order to use EDF instead of FCFS as was proposed in \cite{peha1996cost}. We mark this new version as $c\mu/\theta E$. However, as we will show later our proposed technique outperforms this variation as well.

\section{Maximum Utility with Dropping scheduling policy}
\label{sec:MUD}
We next present the Maximum Utility with Dropping (MUD) scheduling policy. The queuing model assumes \ref{B1}-\ref{B5}. We use the notation $\hat{S}_{t}^{M}, \hat{Q}_{t}^{M},..$ to present the MUD parameters and $\hat{S}_{t}^{E}, \hat{Q}_{t}^{E},..$ to present the EDF parameters.

The MUD policy combines both scheduling and dropping mechanisms. The scheduling algorithm is composed of two steps:
\begin{enumerate}
\item $\mbox{MUD}_a$ - Upon arrival the policy inserts the new job to the queue keeping the EDF order. If the insertion will cause a job to miss its deadline, the policy drops the packet with the minimal throughput ratio ($\frac{w_i}{b_i}$) from the queue.
\item $\mbox{MUD}_s$ - Upon terminating the service of a job the policy processes one of the two jobs in the head of the queue. The selected job is the one with the highest throughput ratio as long as it does not cause the second job to miss its deadline similarly to MEDF.
\end{enumerate}

Let $o(J_i, t)$ and $s(J_i, t)$ be functions that characterize the queue potential order of $J_i$ with policy MUD at time $t$.
\begin{itemize}
	\item $o(J_i ,t)$ is the index of job $J_i$ in $Q_t^{M}$. $o(J_i, t)=1$ means that $J_i$ is at the head of the queue.
	\item $s(J_i, t):= \sum\limits_{J_k:o(J_k, t) <o(J_i, t)}{b_k}$ is the time a job waits until it is processed assuming that no new jobs arrive until it starts processing.
\end{itemize}

Below we describe how the MUD policy handles a job arrival. Let $J_i$ be the new job which reaches to the queue $Q_{t_{i-1}}^{M}$  at time $t_i$.
\begin{tcolorbox}
\center{$\mbox{MUD}_a$}
\begin{enumerate}
	\item Wait for arrival of a new job ($J_i$)
	\item If  $Q_{t_i}^{M} = \emptyset$ and the server is idle then process $J_i$ and go to statement 1
	\item Add $J_i$ to the queue according to shortest time to expiry order. If there are already jobs with the same expiration time, order them in a descending order of their rewards
 	\item Find the first job ($J_k$) which misses its deadline due to the insertion of $J_i$ into the queue $( o(J_k, t_i) \geq o(J_i, t_i))$
 	   \begin{equation*}
		o_k :=
			\begin{cases}
				\underset{e_j< s(J_j, t_i) }\min(o(J_j, t_i)) &\text{ $\exists j: e_j \leq s(J_j, t_i)$ }
				\\
				\infty &\text{otherwise.}
			\end{cases}
	  \end{equation*}
	\item If $o_k := \infty$ then go to statement 1
	\item \label{mud5} Find the job ($J_l$) with the minimum reward per service time. If there are several pick the one with the shortest time to expiry. $J_l := \underset{o(J_j, t_i) \leq o_k} {argmin}(\frac{w_j}{b_j})$
	\item Drop job $J_l$ from the queue
	\item Go to statement 1
\end{enumerate}

\end{tcolorbox}

Note that $o(J_k, t_i)$ and $s( J_k, t_i)$ values change after adding a new job $J_i$ at time $t_i$  as follows  if $o(J_k, t_i) > o(J_i, t_i)$ then $s(J_k, t_i) := s(J_k, t_{i-1}) + b_i$ and $o(J_k, t_i) = o(J_{k}, t_{i-1})+1$ else there is no change. Different presentation of the policy is that it marks the jobs to drop at the arrival time and postpone the dropping stage to the time to the job reaches the head of the queue.

\begin{tcolorbox}
\center{$\mbox{MUD}_s$}
\begin{enumerate}
	\item Wait until $|Q_j^{M}| \ge 2$ and the server is idle. Assume the time is $t_j$.
	\item $J_i := \underset{J_k \in Q_j^M} {argmin} (e_k)$
	\item $J_l := \underset{J_k \in Q_j^M-\{J_i\}} {argmin} (e_k)$
	\item If $(e_i < t_j)$ then drop $J_i$
	\item else if $(e_i \geq (t_j + b_l)) \wedge (\frac{w_j}{b_j} < \frac{w_l}{b_l})  $ then process $J_l$
	\item else process $J_i$
	\item Return to state 1
\end{enumerate}
\end{tcolorbox}

Next we analyze the complexity of adding a new job to a queue in the MUD policy. The MUD policy uses its dropping policy to maintain the queue in such a way that $\hat{Q_{t}}^{M} = Q_{t}^{M}$. According to formula \ref{Qbound} we can bound the queue size. Adding a job to the queue and dropping jobs from the queue require one sequential pass over the queue. Since the ratio $\frac{w_m}{b_m}$ is calculated upon arrival the policy only needs to compare  $\frac{w_m}{b_m}$ and calculate the $s_i$.
From the above we can conclude that the complexity of adding a job to the queue is bounded by the size of the queue potential  as defined in \ref{Qbound}.

In the FIFO policy the complexity of adding a new job is O(1). In the EDF policy the complexity of adding a new job is $O(\log(\max{|\hat{Q_t^{F}}|}))$ and in the CBS policy the complexity of adding a new job is $O(\max{|\hat{Q_t^{C}}|})$ with a calculation of the exponent and multiplication as described in \cite{peha1996cost}.
The complexity of adding a new job to the queue in the MUD is larger than EDF but more efficient compared to the other policies. Using more complicated data structure the complexity can be reduced.

\section{MUD performance analysis}  
\label{sec:perfanal}
\HighLight{ The following section presents the MUD performance analysis. This section analyzes the MUD scheduling policy performance in the $GI^+/D/1-GI^+/GI^+$ case. Then it proves the policy optimality in the $GI^+/D/1-GI^+/GI^/B+$ case.
}
\subsection{Deterministic Time Case}  
\label{sec:MUDdeter}
\HighLight{
The analysis of the policies performance assumes that the queue becomes empty infinitely often. The assumption below forces this behavior.
\begin{enumerate}[label=A\arabic*]
\setcounter{enumi}{\value{AssumptionsCounter}}
	\item \label{A2} $\Bbb P(A_i > D_{max}) > 0$.
\setcounter{AssumptionsCounter}{\value{enumi}}
\end{enumerate}
Assumption \ref{A2} implies that there are infinitely many times where the queue is empty irrespective of the scheduling policy, i.e, for any two policies $\pi_1$ and $\pi_2$ there are $\tau_1, \tau_2,...\tau_{n},..$ times ($\tau_{n-1} < \tau_{n}$) which satisfy that both queues are empty; i.e. $Q_{\tau_n}^{\pi_1}=Q_{\tau_n}^{\pi_2}=\emptyset$.
}

\begin{definition}\label{D2A}
If $\forall n: U_{\tau_n}^{\pi_1} = U_{\tau_n}^{\pi_2}$ then policies $\pi_1$ and $\pi_2$ are equivalent.
\end{definition}

\begin{definition}\label{D2}
If $\Bbb P(\lim\limits_{n \rightarrow \infty} ((U_{\tau_n}^{\pi_1} - U_{\tau_n}^{\pi_2}) \geq 0) = 1$ then $\pi_1$ is as good as $\pi_2$
\end{definition}

\begin{definition}\label{D3}
If $\Bbb P(\lim\limits_{n \rightarrow \infty} (U_{\tau_n}^{\pi_1} - U_{\tau_n}^{\pi_2}) = \infty) = 1$ then $\pi_1$ is better than $\pi_2$
\end{definition}

\begin{definition}\label{D4} 
If for any policy $\pi_2$, $\pi_1$ is as good as policy $\pi_2$ then $\pi_1$ is optimal.
\end{definition}


In the following section we analyze MUD performance against the group of policies that transmits the maximal number of packets. We show that MUD policy performance and the EDF policy performance in a $GI^+/D/1-GI^+/D$ is equal.  In this model the service times and the rewards are deterministic (one unit) which is equivalent to the $GI^+/D/1-GI$ queuing model and the cumulative reward function is the number of processed jobs. In \cite{panwar1988optimal} it was shown that EDF policy is optimal in the terms of number of jobs to be processed for a queuing model  satisfying $GI^+/D/1-GI$. We continue the analysis by showing that in a queuing model satisfying $GI^+/D/1-GI^+/GI^+$, the MUD policy is as good as the EDF policy or similar policies. We conclude with additional  assumptions  to show that the MUD policy is better than the EDF policy or similar policies.

\begin{lemma} \label{L24A}
For any queuing model satisfying $GI^+/D/1-GI^+/GI^+$, $\forall t: |Q_t^{M}| = |Q_t^{E}|$.
\end{lemma}

\begin{proof}
Since the lemma is related solely to the number of jobs that are processed, the rewards are not used in the proof.
The proof is based on the inequality between the two sides.

The first step is to prove that $\forall t: |Q_t^{M}| \leq |Q_t^{E}|$ .This inequality is derived from the optimality of EDF in a queuing model  satisfying $GI^+/D/1-GI$ as was proven in \cite{panwar1988optimal}.

For the other side $\forall n: |Q_{\hat{t}_n}^{M}| \geq |Q_{\hat{t}_n}^{E}|$ the proof is achieved by induction on $\hat{t}_n$. Let $t=\hat{t}_n$, the induction hypothesis is $|Q_{t}^{E}| = |Q_{t}^{M}|$ and the two policies provide service simultaneously (synchronized).

By definition at $n=0$ the queues are empty and the servers are idle $(Q_{t}^{E} = Q_{t}^{M} = \phi)$. At $n=1$ the first job arrives and is processed immediately by both policies, i.e., they are synchronized. This event is similar to all events when the queues are empty and the server is idle (\ref{QE1}).
By definition the service time is deterministic; thus the policies are synchronized as long as $Q_t^{E}$ and $Q_t^{M}$ have the same size.
For the induction step, assume that the processing start time is synchronized and $|\hat{Q}_{t}^{M}| = |Q_{t}^{E}| = k$.
We need to show that fot $t'=\hat{t}_{n+1}$,  $|Q_{t'}^{M}| = |Q_{t'}^{E}|$ is kept, and the service start time is synchronized in the following cases:

\begin{itemize}
\item Begin processing a job in one or both policies (\ref{QE5})
\item Arrival of a new job when the server is busy (\ref{QE2})
\end{itemize}

First we analyze the case of \ref{QE5}. Let $|Q_{t}^{E}| = |Q_{t}^{M}| > 0$.
When the event of process begins $\Delta Q_{t'}^{E} = \Delta Q_{t'}^{M} = -1$ . The processing begins simultaneously (induction hypothesis) and since the processing time is deterministic, both policies end processing simultaneously , thus adhering to the induction hypothesis.

Second we analyze the case of \ref{QE2}. As stated above \ref{QE2} contains two type of events: \ref{QE3} when the new job is added to the queue potential and the queue potential is growing and \ref{QE4} when the new job causes one of the jobs in the queue potential to miss its deadline.

The arrival of a new job does not impact the synchronization of the processing. We analyze events \ref{QE3} and \ref{QE4} in both  policies.

\begin{enumerate}
\item Similar behavior of the two policies in the event \ref{QE3}.
\item Similar behavior of the two policies in the event  \ref{QE4}.
\item $\Delta Q_{t'}^{E}=0, \Delta Q_{t'}^{M} = 1$, MUD event is \ref{QE3} while the EDF event is \ref{QE4}.
\item $\Delta Q_{t'}^{E}=1, \Delta Q_{t'}^{M} = 0$, MUD event is \ref{QE4} while the EDF event is \ref{QE3}.
\end{enumerate}
The first two cases preserve the induction hypothesis as required. The third case contradicts the optimality of the EDF; in other words the event cannot exist. We show that the fourth event cannot exist.

By definition $Q_{t}^{M} \subseteq \hat{Q}_{t}^{E}$ since they have the same feed of jobs and the queues are ordered by expiration time.
Let $J_j$ be the job that is dropped by the MUD policy at time $t$. If $J_j \in Q_{t}^{E}$,
 then, there are more jobs in $Q_{t}^{M}$ then in $Q_{t}^{E}$ whose expiry time is less than $e_j$. This  contradicts the optimality of EDF (please note that the order of the queues is the same). If $J_j \notin Q_{t}^{E}$  since $e_j \geq e_i$ we can add it to $Q_{\hat{t}_{n-1}}^{E}$ without dropping and have a larger queue since $J_i$ is added at time $t_i$. This contradicts the optimality of EDF. This means that case 4 cannot exist.
The result is that the processing is synchronized and the queues have a similar size.
\end{proof}

\begin{lemma}  \label{L2}
For any queuing model satisfying $GI^+/D/1-GI^+/D$, the MUD and the EDF policies are equivalent.
\end{lemma}
\begin{proof}
Since the reward is deterministic, without loss of generality assume $W_i=1$. Therefore,  from lemma \ref{L1} and lemma \ref{L24A} we get that  $\forall t : U_t^{E} =|\hat{S}_t^{E}| = |\hat{S}_t^{M}|=U_t^{M}$ thus MUD and EDF are equivalent.
\end{proof}

\begin{theorem}  \label{T10}
For any queuing model satisfying $GI^+/D/1-GI^+$, the MUD policy is optimal in terms of the number of jobs to process where the process time is exactly one unit.
\end{theorem}
\begin{proof}
Queuing models satisfying $GI^+/D/1-GI^+/D$ are equivalent to queuing systems satisfying  $GI^+/D/1-GI$ if we measure the number of jobs that were processed and $W_i = 1$. In the case of a deterministic service time, knowing the service time upon arrival or just before service is equivalent since both are known to be a constant.
In \cite{panwar1988optimal} it was shown that the EDF policy is optimal in terms of the number of jobs to process for the discrete time $GI^+/D/1-GI$ queue where process time is exactly one unit. In lemma \ref{L2} we showed that the MUD is equivalent to EDF; thus MUD is optimal in terms of the number of jobs to process where the process time is exactly one unit and the reward is also one unit.
\end{proof}

Let $\mathcal{M}$ be the group of policies that processes the maximal number of jobs assuming $GI^+/D/1-GI^+/GI^+$. $EDF, MUD \in \mathcal{M}$ as a result of \cite{panwar1988optimal} and Lemma \ref{L2} .
 Let $\mathcal{M}_E \subseteq \mathcal{M}$ be the group of policies that their expected reward is $E(W_i)$.

\HighLight{
In order to prove lemma \ref{LT2} we add the following assumptions: 
}
\begin{enumerate}[label=A\arabic*]
\setcounter{enumi}{\value{AssumptionsCounter}}
	\item \label{A3} $A_i$, $D_i$ and $W_i$ have discrete probability distribution function.
	\item \label{A21} $A_i, D_i$ and $W_i$ are independent of each other.
\setcounter{AssumptionsCounter}{\value{enumi}}
\end{enumerate}
\HighLight{
Assumption \ref{A3} is used to simplify the proofs. In addition $W_i$ has finite presentation in the communication environment forcing discrete distribution. $A_i$ and $D_i$ may have a continuous distribution function as well. 
Assumption \ref{A21} is reasonable since $A_i$ is related to the server's arrival process. $D_i$ is derived from the original deadline reducing different network delays. $W_i$ is generated independently by the data source. 
}


\begin{lemma}\label{LT2}
For any queuing model satisfying $GI^+/D/1-GI^+/GI^+$ and assumptions 
\HighLight{
\ref{A2}-\ref{A21}},  MUD policy is as good as EDF.
\end{lemma}
\begin{proof}
This proof analyzes the behavior of the different types of events.  For purposes of this proof let $\pi \in \{ MUD, EDF \}$ and let $t = \hat{t}_n$ the time of event $\mathcal{E}_n^{\pi})$ and $t'=\hat{t}_{n-1}$ the time of event $\mathcal{E}_{n-1}^{\pi}$ be the times of the current event and the prior event.
The different cases are:
\begin{enumerate}
\item Arrival of a new job at events  \ref{QE1} and \ref{QE2} where $\Delta Q_{t}^{E} \ne \Delta Q_{t}^{M}$. This contradicts Lemma \ref{L24A} since one of the $Q_{t}^\pi$ becomes longer than the other one.
\item \label{L2A}Arrival of a new job where $\Delta Q_{t}^{E} = \Delta Q_{t}^{M} = 1$ (\ref{QE3}). In this case a job was added to both queue potentials without reducing the number of jobs in them. The queue potential reward growth is the same in both queues.
\item \label{L3A} Arrival of a new job where $\Delta Q_{t}^{E} = \Delta Q_{t}^{M} = 0$ (\ref{QE4}). In this case  there is a change in the queue potential rewards which can be different between policies.
\item Processing begins of a job in one or both policies (\ref{QE5}) - Process starting moves the job from $Q_{t'}^\pi$ to $\hat{S}_{t}^\pi$ keeping $S_{t'}^{\pi}=S_{t}^{\pi}$.
\item A job does not meet its deadline (\ref{QE6}) or is dropped (\ref{QE7}) - By definition this job did not belong to $S_{t'}^{\pi}$.
\end{enumerate}

We define $\Delta U_{t}^{\pi}$ explicitly and compare the difference between the two policies. For arrival events \ref{QE1}, \ref{QE3} and \ref{QE4} let $J_i$ be the arriving job and $t=\hat{t}_n=t_i$ be the current time. We use $W_j^{E}$ or $W_k^{M}$ to describe the rewards of jobs $J_j$ and $J_k$ that are removed or dropped from the queue potentials in the different policies.

Let $\Delta U_{t}^{\pi}$ be:
\begin{equation}
\Delta U_{t}^{E} =
		\begin{cases}	
			0 					& 	\text{Events \ref{QE5}, \ref{QE6} and \ref{QE7}}
			\\
			W_i 					& 	\text{Events \ref{QE1}, \ref{QE3}}
			\\
		 	W_i - W_j^{E}		&	\text{Event \ref{QE4}}
		\end{cases}
\end{equation}
\begin{equation}
\Delta U_{t}^{M} =
		\begin{cases}
			0 					& 	\text{Events \ref{QE5}, \ref{QE6} and \ref{QE7}}
			\\
			W_i 					& 	\text{Events \ref{QE1}, \ref{QE3}}
			\\
			W_i - W_k^{M} 	& 	\text{Event \ref{QE4} and } W_k^{M} < W_i
			\\
			 0 					&  	\text{Event \ref{QE4} and } W_k^{M} \geq W_i
		\end{cases}
\end{equation}

We are interested in the difference between the potential rewards,
$\Delta U_{t}^{M}$ and $\Delta U_{t}^{E}$. In the events \ref{QE1}, \ref{QE3}, \ref{QE5}, \ref{QE6} and \ref{QE7}, $\Delta U_{t}^{M}=\Delta U_{t}^{E}$ as was presented above. We focus on the behavior of events of type \ref{QE4} which impacts the queue potential rewards.

Let $\tau = \{\tau_1, \tau_2...\}$ be a sequence of times satisfying \ref{QE4} and let $t_i = \tau_m$.
Let $\hat{\Bbb E}(W_i)$ be the empirical mean of the sequence of  all events of type \ref{QE4} up to time $\tau_m = t_i$.

We begin by analyzing the queue potential behavior in the EDF. In this event an arbitrary job from $Q_{t_{i-1}}^{E} \cup \{J_i\}$ misses its deadline and is removed from the queue potential. The EDF policy orders the queue according to expiration time. Hence the job that misses its deadline only depends on its deadline. Since the deadlines and the rewards are independent 
\HighLight{
according to assumption \ref{A21}
} and rewards are i.i.d  we get:

\begin{equation}\label{eqEDF}
\begin{split}
\lim\limits_{m \rightarrow \infty} \hat{\Bbb E} & (\Delta U_{\tau_m}^{E})
 = \lim\limits_{m \rightarrow \infty} \hat{\Bbb E} (W_i - W_j^E)\\
& = \lim\limits_{m \rightarrow \infty} \hat{\Bbb E} (W_i) - \lim\limits_{m \rightarrow \infty} \hat{\Bbb E} (W_j^E) \\
& \xrightarrow{a.s} \Bbb E (W_i) - \Bbb E (W_j) = 0
\end{split}
\end{equation}
This means that $EDF \in \mathcal{M}_E$.

On the other hand the MUD policy drops a job with the minimal reward from $Q_{t'}^{M} \cup \{J_i\}$ as defined in statements 5-8 of the policy. Let $J_k^M$ be the job that is chosen to be dropped in the MUD case.  Let $n$ be the number of events in sequence $\tau$ which occur  up to event $t_i$.
\begin{equation}
\begin{split}
\hat{\Bbb E} (\Delta U_{\tau_m}^{M}) &
=\frac{1}{m}\sum\limits_{i \leq m}\Delta U_{\tau_{\hat{i}}}^{M}\\&
=\frac{1}{m}(\sum\limits_{\hat{i} \leq m:W_k \geq W_i }\Delta U_{\tau_{\hat{i}}}^{M}+\sum\limits_{\hat{i} \leq m:W_k < W_i } U_{\tau_{\hat{i}}}^{M})\\&
=\frac{1}{m}\sum\limits_{\tau_{\hat{i}} \leq \tau_{m}:W_k < W_i } U_{\tau_{\hat{i}}}^{M} \geq \frac{1}{m}\sum\limits_{\tau_{\hat{i}} \leq \tau_{m}:W_k < W_i } \delta_W \geq 0
\end{split}
\end{equation}

If $\lim\limits_{m \rightarrow \infty} \hat{\Bbb E} (\Delta U_{\tau_m}^{M})  \xrightarrow{a.s} 0$ then the probability of the event that $W_k < W_i$ in the sequence is 0; i.e., that the number of beneficial exchanges achieved by the MUD is negligible. The MUD policy orders the queue by expiration time exactly like EDF; thus, both policies behave identically at probability 1. Alternatively $\lim\limits_{m \rightarrow \infty} \hat{\Bbb E} (\Delta U_{m}^{E})  \xrightarrow{a.s} \epsilon > 0$.

The conclusion is that the MUD policy is at least as good as the EDF policy.
\end{proof}

\begin{theorem}\label{T2}
For any queuing model satisfying $GI^+/D/1-GI^+/GI^+$ and assumptions 
\HighLight{
\ref{A2}-\ref{A21}
}, MUD policy is as good as any policy $\pi \in \mathcal{M}_E$.
\end{theorem}
\begin{proof}
By definition  $\forall \pi \in \mathcal{M}_E:$ their expected reward is $E(W_i)$. From Dertouzos theorem \cite{stankovic2012deadline} we know that EDF is optimal in that if there exists any policy $\pi$ that can build a feasible schedule then EDF also build a feasible schedule. Adding to it lemma \ref{LT2} we deduce that MUD policy is as good as any policy $\pi \in \mathcal{M}_E$.
\end{proof}

\HighLight{
In order to prove that the MUD policy is better than any policy $\pi \in \mathcal{M}_E$ the following assumptions are required:
}

\begin{enumerate}[label=A\arabic*]
\setcounter{enumi}{\value{AssumptionsCounter}}
	\item \label{A4} $\Bbb P( B_i < D_{max}) > 0$.
	\item \label{A41} $\exists \delta >0: \Bbb P(A_i = B_{min}-\delta ) > 0$, Let $A_{\delta} = B_{min} - \delta$
	\item \label{A6} $\Bbb Var( W_i ) \neq 0$.
	\item \label{A8} Let $\delta_W$ be the minimal difference between two different rewards. \ref{A3} and \ref{A6} imply that $\delta_W > 0$.
	\setcounter{AssumptionsCounter}{\value{enumi}}
\end{enumerate}

\HighLight{
Assumption \ref{A4} guarantees that there is always a positive probability that a feasible job arrives. This is a natural assumption, since otherwise the will be no jobs to process. Assumption \ref{A41} guarantees that jobs wait in the queue with positive probability irrespective of the scheduling policy. \ref{A8} always holds when $W_i$ has a discrete distribution with finite support.
}


\begin{lemma}\label{L6}
For any queuing model satisfying $GI^+/D/1-GI^+/GI^+$ and assumptions \ref{A1} - \ref{A8},
there exists an arrival $\hat{\mathcal A}_{j(n+2)}^\pi= \{\ref{QE4}\}$ which occurs infinitely often such that $\Delta U_{t_{j(n+2)}}^{M} = \delta_W$.
\end{lemma}
\begin{proof}
In order to prove the lemma we generate a sequence of job arrivals called $\tilde{\mathcal A}_i^\pi \subset \mathcal A$ with strictly positive probability. We prove that the last event of sequence is of type \ref{QE4} regardless the scheduling policy. 
The sequence of jobs is defined as follows:
\begin{enumerate}
	\item Let $J_i=<a_i,1, D_{max},W_{min}, t_i+D_{max}>$ where $a_i  > D_{max}$ be the first job in the sequence.
\item Let $\{<A_{\delta},1, D_{max},W_{min}, t_k+D_{max}>\} _{k=i+1}^{i+n}$  be a sequence of $n$ jobs with constant parameters. Let $n$ be defined by:
\begin{equation} \label{eq2}
n= \lceil \frac{\ D_{max} }{\delta}\rceil
\end{equation}
\item Let $J_{i+n+1}=<A_{\delta},1, D_{max},W_{max}, t_{i+n+1}+D_{max}>$ be the last job in the sequence.
\end{enumerate}
The first arrival empties the queue potential with positive probability by assumption \ref{A2}. By assumption \ref{A4}, if $a_k = A_{\delta}$ the queue is increased upon arrival of $J_k$. After $n+1$ arrivals there is at least one job waiting in the queue when $J_{i+n+1}$ arrives.  Meaning that one job must miss its deadline ($\Delta Q_{\tau_{i+n+1}}^\pi=0$). 
The probability of the arrival of $J_k \in \{<A_{\delta},1, D_{max},W_{min}, t_k+D_{max}>\} _{k=i+1}^{i+n}$ is $\Bbb P(A_i = A_{\delta}) \Bbb P(D_i = D_{max}) \Bbb P(W_i = W_{min}) > 0 $ due to assumption \ref{A21}. The probability of the arrival of $J_{i+n+1}$ is $\Bbb P(A_i = A_{\delta}) \Bbb P(D_i = D_{max}) \Bbb P(W_i = W_{max}) > 0$.
As a result the event $\hat{\mathcal A}_i^\pi$ has positive occurrence probability since:
\begin{equation} \label{eq2}
\begin{split}
\Bbb P(\tilde{\mathcal A}_k) &=  \Bbb P (A_i > D_{max})[\Bbb P(A_i = A_{\delta}) \Bbb P(D_i = D_{max}) \\
& \Bbb P(W_i = W_{min})]^{n}  \Bbb P(A_i = A_{\delta}) \Bbb P(D_i = D_{max})\\
& \Bbb P(W_i = W_{max}) > 0
\end{split}
\end{equation}
The MUD policy drops a job with a smaller reward than $J_{i+n+1}$ causing the cumulative reward functions to increase at least by $\delta_W$. For comparison, the EDF policy drops $J_{i+n+1}$ because it misses its deadline.

The arrivals are independent and mutually exclusive events by definition (forcing an empty queue) and $\sum\limits_{j=1}^{\infty} \Bbb P(\hat{\mathcal A}_{j(n+2)}^\pi) = \infty$. By the second Borel Cantelli lemma,  $\Bbb P(\limsup\limits_{j \rightarrow \infty}\hat{\mathcal A}_{j(n+2)}^\pi) = 1$; i.e., meaning that this event occurs infinitely often.

\end{proof}
From lemma \ref{L6} the series of events of type $\hat{\mathcal A}_i$ satisfies:
\begin{equation}
\Delta U_{t_{i}}^{M} - \Delta U_{t_{i}}^{E} \geq \delta_W
\end{equation}

\begin{theorem}\label{T3}
For any queuing model $GI^+/D/1-GI^+/GI^+$ satisfying assumptions \ref{A1} - \ref{A8}, MUD policy is better than any policy $\pi \in \mathcal{M}_E$.
\end{theorem}
\begin{proof}
From theorem  \ref{T2} we know that $\Delta U_{t}^{M}=\Delta U_{t}^{E}$ in all cases except \ref{QE4} case and that $\lim\limits_{n \rightarrow \infty} \hat{\Bbb E} (\Delta U_{\tau_n}^{E}) \xrightarrow{a.s} 0$. in case \ref{QE4}.
From lemma \ref{L6} we can conclude as regards \ref{QE4} that
$\lim\limits_{i \rightarrow \infty}\sum\limits_{\tau_i}\Delta U_{\tau_{i}}^{E} = \infty$.
By the strong law of large numbers there exists an $n_0$ that for all $n > n_0: \hat{\Bbb E} (\Delta U_{\hat{t}_n}^{E}) <  \frac{\delta_W \Bbb P(\hat{\mathcal A}_k)}{3}$ with probability 1. Let
$C  = - \sum\limits_{m \leq n_0 :\tau_m \in \tau} U_{\hat{t}_m}^{E}$. Then

\begin{equation}
\begin{split}
U_{\hat{t}_n}^{M}& - U_{\hat{t}_n}^{E}   \geq C +\sum\limits_{ n \geq m \geq {n_0}}(\Delta U_{\hat{t}_m}^{M} - \Delta U_{\hat{t}_m}^{E})\\&
= C + \sum\limits_{n_0 \leq m \leq n:\hat{t}_m \not \in \tau}(\Delta U_{\hat{t}_m}^{M} - \Delta U_{\hat{t}_m}^{E})\\
& \geq C+\sum\limits_{n_0 \leq m \leq n:\hat{t}_m \in \tau}(\Delta U_{\hat{t}_m}^{M} - \Delta U_{\hat{t}_m}^{E})
\end{split}
\end{equation}
\begin{equation}
\begin{split}
\lim\limits_{n \rightarrow \infty} (U_{t_n}^{M} - U_{t_n}^{E}) & \geq  \\
&C+ \lim\limits_{n \rightarrow \infty}\sum\limits_{n_0 \leq m \leq n:\hat{t}_m \in \tau}(\Delta U_{\hat{t}_m}^{M}- \Delta U_{\hat{t}_m}^{E})
\\&
\geq  C+ \lim\limits_{n \rightarrow \infty}(n-n_0) \frac{\delta_W \Bbb P(\hat{\mathcal A}_k)}{3}
\xrightarrow{a.s}\infty
\end{split}
\end{equation}
This means that the MUD policy is better than EDF policy when assumptions \ref{A1} - \ref{A8} are held. 
Adding to it Dertouzos theorem \cite{stankovic2012deadline} we can deduce that MUD policy is better than any policy $\pi \in \mathcal{M}_E$.
\end{proof}

\subsection{Deterministic Time and Dual Reward Case }  
\label{sec:MUDimp}
In this section we study the optimality of the MUD policy for deterministic service time. To simplify an already complicated notation we focus on the two priority levels case. We define a  dual reward process by $W_i \in \{w_{min}, w_{max}:w_{min} < w_{max}\}$ i.e. $(W_i \sim \mathcal{B}(p))$. 

In the dual reward case we can split the queue into two sub-queues. Let 

\begin{itemize}

\item $\HighLight{Q_i^{M(max)}}=\{w_j \in Q_i^M:w_j=w_{max}\}$ be the set of jobs in the queue potential with maximal reward.
\item $\HighLight{Q_i^{M(min)}}=\{w_j \in Q_i^M:w_j=w_{min}\}$  be the set of jobs in the queue potential with minimal reward.
\end{itemize}
Similarly, we split the group of processed jobs to two groups:
\begin{itemize}
\item $\HighLight{\hat{S}_i^{M(max)}}=\{w_j \in \hat{S}_i^M:w_j=w_{max}\}$ be the set of processed jobs with maximal reward.
\item $\HighLight{\hat{S}_i^{M(min)}}=\{w_j \in \hat{S}_i^M:w_j=w_{min}\}$ be the set of processed jobs with maximal reward.
\end{itemize}

Let $R_i^{M(max)}=|Q_i^{M(max)} \cup \hat{S}_i^{M(max)}|$ be the number of jobs with maximal reward in the queue potential and in the group of processed jobs.

We say that $R_i^{M(max)}$ is optimal if no other policy can process more jobs with maximal rewards than $R_i^{M(max)}$. 
For the proofs in this section we assume the following:

\HighLight{
One of the major problems in non-preemptive scheduling without idle time is the problem of priority inversion \cite{stankovic2012deadline}. To overcome this problem we add the following assumption:
\begin{enumerate}[label=A\arabic*]
	\setcounter{enumi}{\value{AssumptionsCounter}}
	\item \label{A9} Let $D_{min} > 2B_{max}$ meaning that the deadline is at least twice the time required for processing the longest packet.
\end{enumerate}
For non-border routers in a network this assumption is reasonable since it allows the packets to reach their destination.
}

First we show that the mild inversion of two packets at the Queue head in $\mbox{MUD}_s$, does not cause drops of other packets.
\begin{lemma}\label{L7}
Let $J_l, J_k \in Q^M_i:o(J_l, t_i) = 1$ and $o(J_k, t_i) = 2$ and assume \ref{A9}, then replacing the service order of $J_l, J_k$ does not impact the MUD performance.
\end{lemma}
\begin{proof}
In order to prove the lemma we need to show that for all packets that their order in the queue is higher than 2 ($o(J_i, t)  > 2$) their expected service time $s(J_i, t)$ is kept. By definition $s(J_i, t):= \sum\limits_{J_k:o(J_k, t) <o(J_i, t)}{b_k} =\sum\limits_{J_k:2 < o(J_k, t) <o(J_i, t)}{b_k}+\sum\limits_{J_k:o(J_k, t) \leq 2}{b_k}$. Since $\sum\limits_{J_k:o(J_k, t) \leq 2}{b_k}$ is kept regardless the processing order the $s(J_i, t)$ is kept. Assumption \ref{A9} eliminates the possibility that the replacement causes packet to miss its deadline due to a new arrival. This means that assuming \ref{A9} the replacement in the head of the queue does not impact the MUD performance.
\end{proof}

\begin{lemma}\label{L8}
For any queuing model $GI^+/D/1-GI^+/GI^/B+$ satisfying assumption \ref{A9}, $R_i^{M(max)}$ is optimal.
\end{lemma}

\begin{proof}
The main idea behind the proof is that under MUD $Q_i^{M(max)}$ behaves like EDF queue of maximal reward jobs. Jobs with minimal reward are processed by MUD only if they do not impact jobs with maximal reward and close to their expiration time. Since we already proved that MUD processes the maximal number of packets as well as the maximal number of packets with highest priority, it must be optimal.
We wish to prove by induction that $R_i^{M(max)}$ is optimal. For $t_0$ it is true since $R_i^{M(max)}=0$.
Assume by induction that $R_{i-1}^{M(max)}$ is optimal, we wish to prove that it is optimal for $R_i^{M(max)}$, we prove it per  event type;
\begin{enumerate}
\item Event \ref{QE1} the job is processed in all policies. Since no other policy can do better, $R_i^{M(max)}$ is optimal ($R_i^{M(max)}=R_{i-1}^{M(max)}+1$).
\item Event \ref{QE3} the job is added to the queue potential without a drop. Since no other policy can do better, $R_i^{M(max)}$ is optimal  ($R_i^{M(max)}=R_{i-1}^{M(max)}+1$).
\item Event \ref{QE4}. The event happens also in any other policy as well due to the optimality of the MUD by theorem \ref{T10}. Let $J_i$ be the arriving job. The event is divided into two cases:
\begin{itemize} 
\item MUD marks for dropping a job with minimal reward - If the new job has a maximal reward then $R_i^{M(max)}=R_{i-1}^{M(max)}+1$ otherwise $R_i^{M(max)}=R_{i-1}^{M(max)}$. Since no other policy can do better, $R_i^{M(max)}$ is optimal. 
\item MUD marks for dropping a job with maximal reward - this happens only if the new job has a maximal reward.
Assume by negation that exists a policy that drops job with minimal reward instead. We can take it's queue potential and MUD queue potential and extract out of them a new queue potential which has more jobs with maximal reward than the induction assumption. This contradict the induction assumption. Since no other policy can do better, $R_i^{M(max)}$ is optimal.
\end{itemize}
\item Event \ref{QE5} this event does not impact $R_i^{M(max)}$ directly but through event \ref{QE4}. MUD postpones as much as possible giving service to jobs with minimal reward in order to maximize the event of dropping job with minimal reward at event \ref{QE4}. Lemma \ref{L7} states that the replacement does not impact the queue potential assuming \ref{A9}. 
\item Events \ref{QE6} and \ref{QE7} do not impact $R_i^{M(max)}$.
\end{enumerate}
We proved that $R_{i}^{M(max)}$ is optimal. 
\end{proof}

From lemma \ref{L8} we can conclude that at times $Q_{i}^{M}=\emptyset$, $\hat{S}_i^{M(max)}$ is optimal.

\begin{theorem}\label{T4}
For any queuing model $GI^+/D/1-GI^+/B$ satisfying assumption \ref{A9},  MUD policy is optimal. 
\end{theorem}

\begin{proof}
To prove the theorem we need to show that at times $Q_{i}^{M}=\emptyset$, $U^M_{t_i}$ is optimal .

Let $N_i^M$ be the number of jobs process by policy MUD until time $t_i$. 
\begin{equation}
U_{t_i}^M= \sum\limits_{J_l \in S_{t_i}^M}^{}w_l = R_i^{M(max)}W_{max}+(N_i^E-R_i^{M(max)})W_{min}
\end{equation}
In theorem \ref{T10} we proved that $N_i^M$ is optimal and in lemma \ref{L8} we proved that $R_i^{M(max)}$ is optimal. 
Hence $U_{t_i}^M$ includes both the maximal number of packets and maximal number of packets with highest rewards. This implies that no other policy can do better.
\end{proof}

\begin{figure}
  \centering
  \includegraphics[width=3.5in]{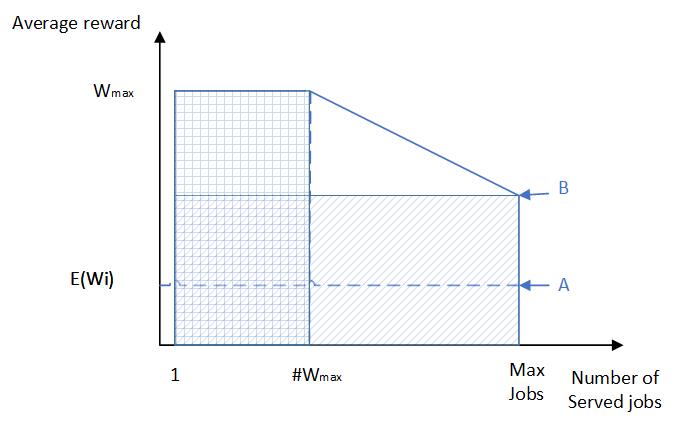}
  \caption[center]{Pareto optimality graph -\\  
   Jobs counting is ordered by descending rewards}
\label{Graph002}
\end{figure}

Figure \ref{Graph002} depicts the behavior of different policies in respect to the number of served jobs and the estimated reward. Point A presents the performance of policies in $\mathcal{M}_E$. MUD performance is between point $A$ and $B$. In the dual reward case we proved that MUD performance is at point B. The area in the graph depicts the total reward accumulated by the policy. We prove that MUD is optimal in the dual reward case concluding that its rectangle has the maximal area. 


Similar analysis shows the following theorem for two packet sizes with equal rewards:

\begin{theorem}\label{T5}
For any queuing model $GI^+/B/1-GI^+/D$ satisfying \ref{A9}, MUD policy is optimal. 
\end{theorem}

This result has both practical and theoretical interest. While EDF is not optimal for non-deterministic service time, we can extend the optimality of MUD to the case of dual service times when there are no rewards. Furthermore, in practice in many cases we have data packets of fixed length (typical in Ethernet) and acknowledgment packets which are much shorter. Assuming a sufficient reception window for ACK packets, the MUD policy is indeed optimal.

\subsection{Non-deterministic Service Time}
The problem of ordering non-deterministic service time queues without preemption is similar to be optimization problems which their complexity is $NP-hard$ \cite{brucker2007scheduling}. 
In addition, if avoiding assumption \ref{A9} policies that do not force idle cannot achieve optimality due to priority inversion problem \cite{stankovic2012deadline}. 
In the next section we present simulation of the non-deterministic service time and compare the policies we discussed.

\section{Numerical Results} \label{NR}
\label{sec:num}
In this section we present the numerical results of simulations comparing the the performance of EDF, CBS, MUD, two versions of $c\mu_k/\theta_k$ and $Greedy$ policy. For the $c\mu_k/\theta_k$ we allocated a queue per reward. We used linear cost function that is based on the queue reward (4 or 10) multiplied by the length of the queue. Since the $c\mu_k/\theta_k$ policies do not take advantage of all at the knowledge of the deadline we introduced a new version of the $c\mu_k/\theta_k$ that instead of ordering the queue by FCFS it orders the jobs using EDF and called it $c\mu/\theta_E$. We also compare the results with a greedy algorithm that process the job with the highest reward.
In the following simulation we tested the M/M/1-M/B case. We ran 100,000 samples with the following parameters:
\begin{itemize}
	\item $A$ - is exponentially distributed with $\lambda_a = 0.9, 1.2, 1.5, 1.8, 2.1$
	\item $B$ - is exponentially distributed with $\lambda_b =1$
	\item $D$ - is exponentially distributed with $\lambda_d =0.005$
	\item $E$ - is distributed $\mathcal{B}(0.5)$ with probability $p$ the reward is 10 and with probability $1-p$ the reward is 4.
\end{itemize}

Figure \ref{RewardDual1Rewards} presents the reward collection of the different policies. Figure \ref{RewardDual1Jobs} presents the number of jobs that were processed. The results are shown relative to the baseline of the EDF policy. 

\begin{figure}
  \centering
  \includegraphics[width=3.5in]{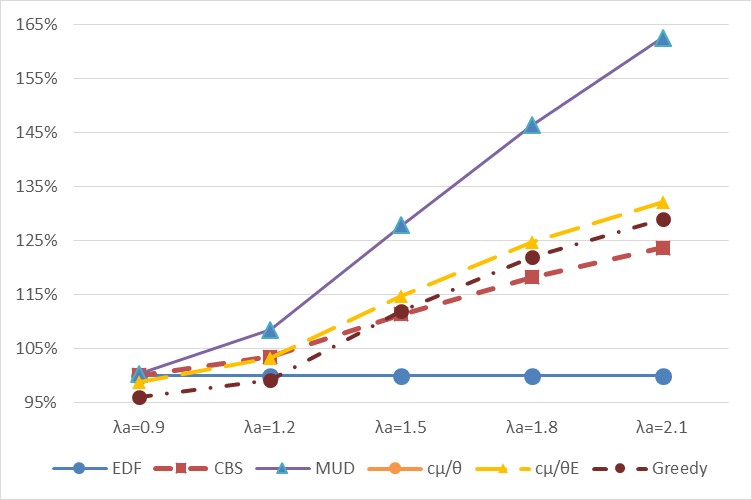}
  \caption[center]{Collected rewards M/M/1-M/B case}
\label{RewardDual1Rewards}
\end{figure}

\begin{figure}
  \centering
  \includegraphics[width=3.5in]{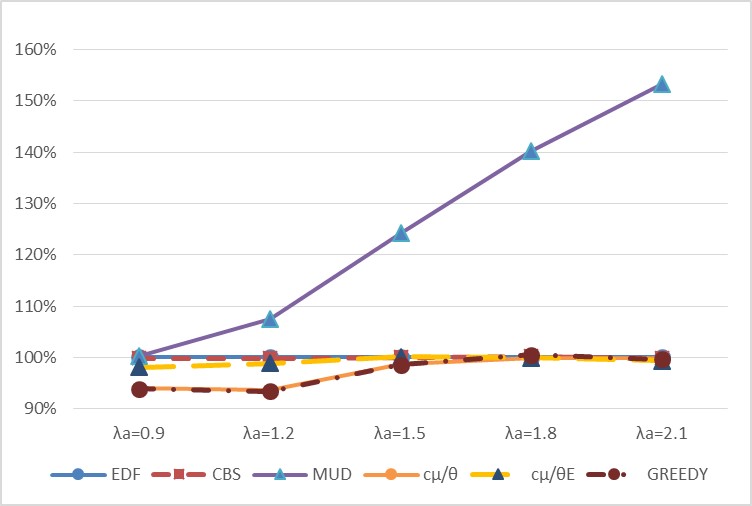}
  \caption[center]{Number of processed Jobs M/M/1-M/B case}
\label{RewardDual1Jobs}
\end{figure}

In the next simulation we tested the M/M/1-M/M case. We ran 100,000 samples with the following parameters:
\begin{itemize}
	\item $A$ - is exponentially distributed with $\lambda_a = 0.9, 1.2, 1.5, 1.8, 2.1$
	\item $B$ - is exponentially distributed with $\lambda_b =1$
	\item $D$ - is exponentially distributed with $\lambda_d =0.005$
	\item $E$ -  is exponentially distributed with $\lambda_e =0.1$.
\end{itemize}

Figure \ref{RewardDual1Rewards} presents the reward collection of the different policies. Figure \ref{RewardDual1Jobs} presents the number of jobs that were processed. The results are shown relative to the baseline of the EDF policy.

\begin{figure}
  \centering
  \includegraphics[width=3.5in]{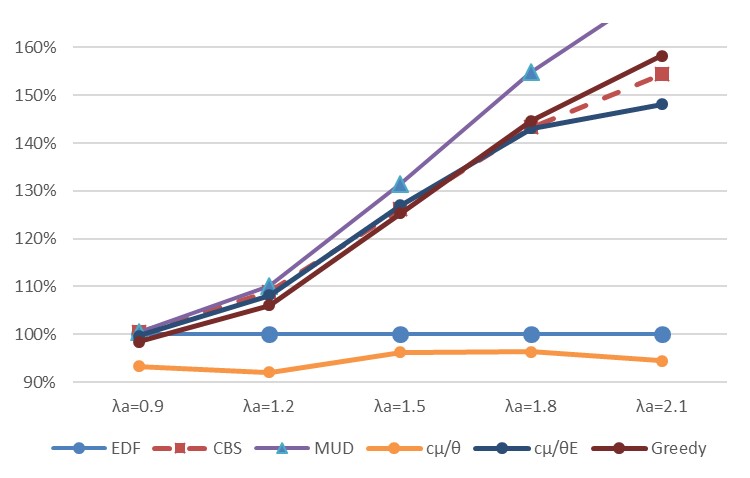}
  \caption[center]{Collected rewards M/M/1-M/M case}
\label{MMMRewards}
\end{figure}

\begin{figure}
  \centering
  \includegraphics[width=3.5in]{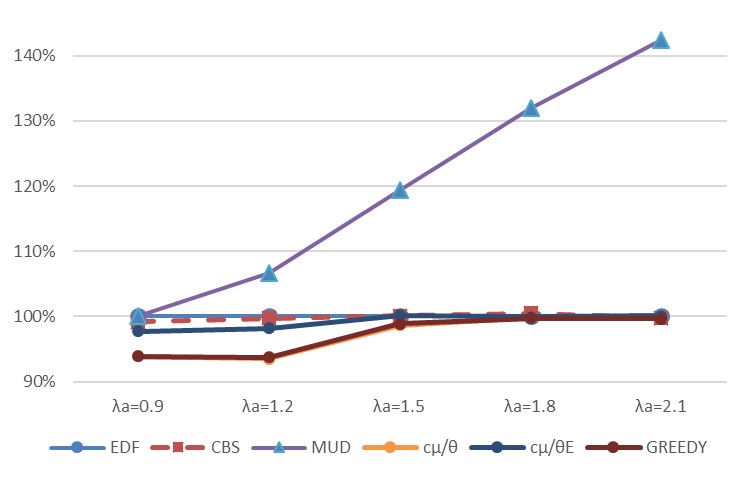}
  \caption[center]{Number of processed Jobs M/M/1-M/M case}
\label{MMMJobs}
\end{figure}
In both simulations we clearly see that MUD performs better than other policies. It processes many more packets and the total collected reward is significantly higher than other policies.
We see also that MUD outperforms when the number of overrun increases due to higher ration between the arrival rate and the service rate ($\rho$). 

\section{Conclusion}
\label{sec:conclusion}
We presented a novel scheduling policy that takes advantage of knowing the job parameters upon arrival.
We proved that the policy is optimal under deterministic service time and binomial reward distribution. In the more general case we proved that the policy processes the maximal number of packets while collecting rewards higher than the expected reward. We present simulation results that show its high performance in more generic cases while keeping computational demands low.

\HighLight{
}
\appendix
\subsection{Queue Behavior Analysis}
\label{AnnexA}

In this annex a queue behavior analysis of a simple model satisfying $GI^+/D/1-GI^+/GI^+$ is provided.
Assume that the current time is $t=t_i=\hat{t}_n$ and $t'=\hat{t}_{n-1}$.
\begin{enumerate}[label=QE\arabic*]
	\item \label{QE1} Job $J_i$ reaches the system while the server is idle and the queue is empty - In this case $J_i$ is processed  immediately. The queue and the queue potential are kept empty $\hat{Q}_{t}^\pi = Q_{t}^\pi = \emptyset$ and  $\Delta Q_{t}^{\pi} = 0$. $S_{t}^\pi = \hat{S}_{t}^\pi = S_{t'}^\pi \cup \{J_i\} $. $\Delta U_{t}^\pi = w_i$ and $U_{t}^\pi = U_{t'}^\pi + w_i$.
	\item \label{QE2} Job $J_i$ reaches the system while the server is busy - In this case $J_i$ waits in the queue, $\hat{Q}_{t}^\pi =\hat{Q}_{t'}^\pi \cup \{J_i\}$ and $\hat{S}_{t}^\pi = \hat{S}_{t'}^\pi$. The arrival of a job may impact the queue potential $Q_{t}^\pi$ differently according to the parameters of the job and the way the policy arranges the queue. This can cause the following two events:
\begin{enumerate}
	\item \label{QE3} If $\Delta Q_{t}^{\pi}=1$ the job joins the queue potential, $Q_{t}^\pi =Q_{t'}^\pi \cup \{J_i\}$ and $\Delta U_{t}^\pi = w_i$ .
	\item \label{QE4} If $\Delta Q_t^{\pi}=0$, i.e., the queue potential size does not change. This is a result of either not adding the arriving job to the queue potential or removing a different job from the queue potential because it missed its deadline. Assume $J_j; j \leq i$ is the job removed from the queue potential, then $Q_{t}^\pi =Q_{t'}^\pi \cup \{J_i\} - \{J_j\}$ and $\Delta U_{t}^\pi = w_i - w_j$ .

	\item \label{QE4A} Job $J_i$ reaches the system while the server is busy and $\Delta Q_{\hat{t}_n}^{\pi} < 0$. This event cannot take place since the service time is deterministic and there is only a replacement of jobs.
\end{enumerate}
	\item \label{QE5} Job $J_j$ where $j < i$ begins process. In this case  $\hat{Q}_{t}^\pi =\hat{Q}_{t'}^\pi - \{J_j\}$, $Q_{t}^\pi =Q_{t'}^\pi - \{J_j\}$, $\Delta Q_{t}^\pi = -1$, $S_{t}^\pi =S_{t'}^\pi$, $\hat{S}_{t}^\pi =\hat{S}_{t'}^\pi \cup \{J_j\}$ and $\Delta U_{t}^\pi = 0$.

	\item \label{QE6} Deadline expired for $J_j$ - This event does not impact the variables that describe the system. $J_j$ is  kept in the queue. $Q_{t}^\pi$ is not impacted since $J_j \notin Q_{t}^\pi$. $\hat{S}_{t}^\pi = \hat{S}_{t'}^\pi$ as well as $S_{t}^\pi = S_{t'}^\pi$.
	\item \label{QE7} Drop of job $J_j$. The drop means that the job was removed from the queue i.e  $\hat{Q}_{t}^\pi =\hat{Q}_{t'}^\pi  - \{J_j\}$. If  $J_j \in Q_{t'}^\pi$ then $\Delta Q_{t}^\pi = -1$ and $\Delta U_t^\pi = -w_j$.
\end{enumerate}

Next we prove the relationship between $Q_t^\pi$ and $\hat{S}_{t'}^\pi (t' \leq t)$.

\begin{lemma} \label{L1}
For any queuing model satisfying $GI^+/D/1-GI^+/GI^+$, two policies $\pi_1, \pi_2$ and times  $t=\hat{t}_n$, $t'=\hat{t}_m$ ($\hat{t}_n$ and $\hat{t}_m$ are the times of $\mathcal{E}_n^{\pi_1}$ and $\mathcal{E}_m^{\pi_1}$ )
:  if $\forall t, t'$ and $t \geq t': |Q_{t'}^{\pi_1}| \geq |Q_{t'}^{\pi_2}|$  then $|\hat{S}_{t}^{\pi_1}| \geq |\hat{S}_{t}^{\pi_2}|$.
\end{lemma}

\begin{proof}
Assume by negation that there exists $t$ such that $\forall t' \leq t :|Q_{t'}^{\pi_1}| \geq |Q_{t'}^{\pi_2}|$ but $|\hat{S}_{t}^{\pi_1}| < |\hat{S}_{t}^{\pi_2}|$ . We take the minimal $t$ that fulfills the above condition; i.e. $|\hat{S}_{t''}^{\pi_1}|\geq |\hat{S}_{t''}^{\pi_2}|$ but  $|\hat{S}_{t}^{\pi_1}| < |\hat{S}_{t}^{\pi_2}|$ where $t''=t_{n-1}$.
We analyze the behavior by cases:
\begin{itemize}
	\item Job arrives when the queue is empty and the server is idle (\ref{QE1}) then; $|\hat{S}_{t}^{\pi_1}|=|\hat{S}_{t''}^{\pi_1}|+1 \geq |\hat{S}_{t''}^{\pi_2}|+1=|\hat{S}_{t}^{\pi_2}|$ which contradicts the assumption. The process is kept synchronized.
	\item Job arrives when the server is busy (\ref{QE2}). The job is added to the queue and there is no change in the size of $|\hat{S}_t^\pi|$ which contradicts the assumption.
      \item $\pi_1$ and $\pi_2$ are synchronized and they start processing jobs (\ref{QE5}). Then $|\hat{S}_{t''}^{\pi_1}|+1 = |\hat{S}_{t}^{\pi_1}| \geq |\hat{S}_{t}^{\pi_2}| =|\hat{S}_{t''}^{\pi_2}| +1$  which contradicts the assumption.
	\item Only $\pi_1$ starts processing a job. Then $|\hat{S}_{t''}^{\pi_1}| + 1 = |\hat{S}_{t}^{\pi_1}| > |\hat{S}_{t''}^{\pi_2}| = |\hat{S}_{t}^{\pi_2}| $ which contradicts the assumption. The jobs processing is kept synchronized.
	\item Only $\pi_2$ starts processing a job. The case that $\pi_1$ and $\pi_2$ are not synchronized is due to the fact that one of the queues was empty before the other. The assumption that $|Q_{t'}^{\pi_1}| \geq |Q_{t'}^{\pi_2}|$ forces that $|Q_{\hat{t}_k}^{\pi_2}|$ was empty before $|Q_{\hat{t}_k}^{\pi_1}|$ for one or more $\hat{t}_k$-s ($\hat{t}_k$ is the time of event $\mathcal{E}_k^{\pi_1}$). In this case $|\hat{S}_{t}^{\pi_1}| > |\hat{S}_{t}^{\pi_2}|$ and adding a new job to $\hat{S}_i^{\pi_2}$ can change the relationship back to $|\hat{S}_{t}^{\pi_1}| \geq |\hat{S}_{t}^{\pi_2}|$, contradicting the assumption.
\end{itemize}
\end{proof}

\subsection{List of Acronyms}
\begin{acronym}
\acro{CBD}{Cost Based Dropping}
\acro{CBS}{Cost Based Scheduling}
\acro{DifServ}{Differentiated Services}
\acro{EDF}{Earliest Deadline First}
\acro{FCFS}{First Come First Serve}
\acro{FTP}{File Transfer Protocol}
\acro{IBS}{Impatience upon Beginning of Service}
\acro{IES}{Impatience upon End of Service} 
\acro{IETF}{Internet Engineering Task Force}
\acro{IntServ}{Integrated services}
\acro{IoT}{Internet of Things}
\acro{ISP}{Internet Service Provider}
\acro{MDP}{Markov Decision Process}
\acro{MEDF}{Modified Earliest Deadline First}
\acro{MUD}{Maximum Utility with Dropping}
\acro{QoS} {Quality of Service}
\acro{SLA}{Service Level Agreement}
\acro{STE}{Shortest Time to Expiry}

\end{acronym}

\section*{Acknowledgement}
The authors would like to thank professor Lang Tong from Cornell University for useful discussions regarding this manuscript and to the anonymous reviewers whose comments significantly improved the presentation of this paper.

\bibliographystyle{ieeetran}

\begin{thebibliography}{1}

\bibitem{barrer1957queuing1}
D.~Barrer, ``Queuing with impatient customers and indifferent clerks,''
  \emph{Operations Research}, vol.~5, no.~5, pp. 644--649, 1957.

\bibitem{barrer1957queuing2}
------, ``Queuing with impatient customers and ordered service,''
  \emph{Operations Research}, vol.~5, no.~5, pp. 650--656, 1957.

\bibitem{lehoczky1996real}
J.~P. Lehoczky, ``Real-time queueing theory,'' in \emph{Real-Time Systems
  Symposium, 1996., 17th IEEE}.\hskip 1em plus 0.5em minus 0.4em\relax IEEE,
  1996, pp. 186--195.

\bibitem{stankovic2012deadline}
J.~A. Stankovic, M.~Spuri, K.~Ramamritham, and G.~C. Buttazzo, \emph{Deadline
  scheduling for real-time systems: {EDF} and related algorithms}.\hskip 1em
  plus 0.5em minus 0.4em\relax Springer Science \& Business Media, 2012, vol.
  460.

\bibitem{yu2016deadline}
Z.~Yu, Y.~Xu, and L.~Tong, ``Deadline scheduling as restless bandits,'' in
  \emph{Communication, Control, and Computing (Allerton), 2016 54th Annual
  Allerton Conference on}.\hskip 1em plus 0.5em minus 0.4em\relax IEEE, 2016,
  pp. 733--737.

\bibitem{brucker2007scheduling}
P.~Brucker and P.~Brucker, \emph{Scheduling algorithms}.\hskip 1em plus 0.5em
  minus 0.4em\relax Springer, 2007, vol.~3.

\bibitem{bertsekas1992data}
D.~P. Bertsekas, R.~G. Gallager, and P.~Humblet, \emph{Data networks}.\hskip
  1em plus 0.5em minus 0.4em\relax Prentice-Hall International New Jersey,
  1992, vol.~2.

\bibitem{srikant2013communication}
R.~Srikant and L.~Ying, \emph{Communication networks: an optimization, control,
  and stochastic networks perspective}.\hskip 1em plus 0.5em minus 0.4em\relax
  Cambridge University Press, 2013.

\bibitem{katzirscheduling}
R.~C.~L. Katzir, ``Scheduling of voice packets in a low-bandwidth shared medium
  access network,'' \emph{EEE/ACM Transactions on Networking (TON)}, vol.~15,
  no.~4, pp. 932--943, 2007.

\bibitem{briscoe2016reducing}
B.~Briscoe, A.~Brunstrom, A.~Petlund, D.~Hayes, D.~Ros, J.~Tsang, S.~Gjessing,
  G.~Fairhurst, C.~Griwodz, and M.~Welzl, ``Reducing internet latency: A survey
  of techniques and their merits,'' \emph{IEEE Communications Surveys \&
  Tutorials}, vol.~18, no.~3, pp. 2149--2196, 2016.

\bibitem{stankovic1995implications}
J.~A. Stankovic, M.~Spuri, M.~Di~Natale, and G.~C. Buttazzo, ``Implications of
  classical scheduling results for real-time systems,'' \emph{Computer},
  vol.~28, no.~6, pp. 16--25, 1995.

\bibitem{wang2010queueing}
K.~Wang, N.~Li, and Z.~Jiang, ``Queueing system with impatient customers: A
  review,'' in \emph{Service Operations and Logistics and Informatics (SOLI),
  2010 IEEE International Conference on}.\hskip 1em plus 0.5em minus
  0.4em\relax IEEE, 2010, pp. 82--87.

\bibitem{bhattacharya1989optimal}
P.~P. Bhattacharya and A.~Ephremides, ``Optimal scheduling with strict
  deadlines,'' \emph{IEEE Transactions on Automatic Control}, vol.~34, no.~7,
  pp. 721--728, 1989.

\bibitem{peha1996cost}
J.~M. Peha and F.~A. Tobagi, ``Cost-based scheduling and dropping algorithms to
  support integrated services,'' \emph{IEEE Transactions on Communications},
  vol.~44, no.~2, pp. 192--202, 1996.

\bibitem{towsley1991optimality}
D.~Towsley and S.~Panwar, \emph{Optimality of the stochastic earliest deadline
  policy for the G/M/c queue serving customers with deadlines}.\hskip 1em plus
  0.5em minus 0.4em\relax Citeseer, 1991.

\bibitem{panwar1988optimal}
S.~S. Panwar, D.~Towsley, and J.~K. Wolf, ``Optimal scheduling policies for a
  class of queues with customer deadlines to the beginning of service,''
  \emph{Journal of the ACM (JACM)}, vol.~35, no.~4, pp. 832--844, 1988.

\bibitem{larranaga2015asymptotically}
M.~Larra{\~n}aga, U.~Ayesta, and I.~M. Verloop, ``Asymptotically optimal index
  policies for an abandonment queue with convex holding cost,'' \emph{Queueing
  Systems}, vol.~81, no. 2-3, pp. 99--169, 2015.

\bibitem{atar2010cmu}
R.~Atar, C.~Giat, and N.~Shimkin, ``The c$\mu$/$\theta$ rule for many-server
  queues with abandonment,'' \emph{Operations Research}, vol.~58, no.~5, pp.
  1427--1439, 2010.

\bibitem{ayesta2011nearly}
U.~Ayesta, P.~Jacko, and V.~Novak, ``A nearly-optimal index rule for scheduling
  of users with abandonment,'' in \emph{INFOCOM, 2011 Proceedings IEEE}.\hskip
  1em plus 0.5em minus 0.4em\relax IEEE, 2011, pp. 2849--2857.

\bibitem{altman2006analysis}
E.~Altman and U.~Yechiali, ``Analysis of customers’ impatience in queues with
  server vacations,'' \emph{Queueing Systems}, vol.~52, no.~4, pp. 261--279,
  2006.

\bibitem{adan2009synchronized}
I.~Adan, A.~Economou, and S.~Kapodistria, ``Synchronized reneging in queueing
  systems with vacations,'' \emph{Queueing Systems}, vol.~62, no. 1-2, pp.
  1--33, 2009.

\bibitem{kapodistria2011m}
S.~Kapodistria, ``The m/m/1 queue with synchronized abandonments,''
  \emph{Queueing systems}, vol.~68, no.~1, pp. 79--109, 2011.

\bibitem{movaghar1996queueing}
A.~Movaghar, ``On queueing with customer impatience until the beginning of
  service,'' in \emph{Computer Performance and Dependability Symposium, 1996.,
  Proceedings of IEEE International}.\hskip 1em plus 0.5em minus 0.4em\relax
  IEEE, 1996, pp. 150--157.

\bibitem{movaghar2006queueing}
------, ``On queueing with customer impatience until the end of service,''
  \emph{Stochastic Models}, vol.~22, no.~1, pp. 149--173, 2006.

\bibitem{saleh2010comparing}
M.~Saleh and L.~Dong, ``Comparing {FCFS} \& {EDF} scheduling algorithms for
  real-time packet switching networks,'' in \emph{Networking, Sensing and
  Control (ICNSC), 2010 International Conference on}.\hskip 1em plus 0.5em
  minus 0.4em\relax IEEE, 2010, pp. 698--703.

\bibitem{braden1994integrated}
R.~Braden, D.~Clark, and S.~Shenker, ``Integrated services in the internet
  architecture: an overview,'' Tech. Rep., 1994.

\bibitem{blake1998architecture}
S.~Blake, D.~Black, M.~Carlson, E.~Davies, Z.~Wang, and W.~Weiss, ``An
  architecture for differentiated services,'' Tech. Rep., 1998.

\bibitem{kim2013dynamic}
J.~Kim and A.~R. Ward, ``Dynamic scheduling of a {GI/GI/1+ GI} queue with
  multiple customer classes,'' \emph{Queueing Systems}, vol.~75, no. 2-4, pp.
  339--384, 2013.

\bibitem{clare1989value}
L.~Clare and A.~Sastry, ``Value-based multiplexing of time-critical traffic,''
  in \emph{Military Communications Conference, 1989. MILCOM'89. Conference
  Record. Bridging the Gap. Interoperability, Survivability, Security., 1989
  IEEE}.\hskip 1em plus 0.5em minus 0.4em\relax IEEE, 1989, pp. 395--401.

\bibitem{kingman1962queues}
J.~Kingman, ``On queues in heavy traffic,'' \emph{Journal of the Royal
  Statistical Society. Series B (Methodological)}, pp. 383--392, 1962.

\bibitem{kingman2003single}
J.~Kingman and M.~Atiyah, ``The single server queue in heavy traffic,''
  \emph{Operations management: critical perspectives on business and
  management}, vol.~57, p.~40, 2003.

\bibitem{atar2017workload}
R. ~Atar, A ~Lev-Ari, ``Workload-dependent dynamic priority for the
  multiclass queue with reneging,'' 
  \emph{Mathematics of Operations Research}, 2017.

\bibitem{kendall1953stochastic}
D.~G. Kendall, ``Stochastic processes occurring in the theory of queues and
  their analysis by the method of the embedded markov chain,'' \emph{The Annals
  of Mathematical Statistics}, pp. 338--354, 1953.

\bibitem{hopp2008single}
W.~J. Hopp, ``Single server queueing models,'' in \emph{Building
  Intuition}.\hskip 1em plus 0.5em minus 0.4em\relax Springer, 2008, pp.
  51--79.

\bibitem{buyukkoc1985c}
C.~Buyukkoc, P.~Variaya, and J.~Walrand, ``$c\mu$-rule revisited.'' \emph{Adv.
  Appl. Prob.}, vol.~17, no.~1, pp. 237--238, 1985.

\bibitem{bispo2013single}
C.~F. Bispo, ``The single-server scheduling problem with convex costs,''
  \emph{Queueing Systems}, vol.~73, no.~3, pp. 261--294, 2013.

\bibitem{kebarighotbi2009revisiting}
A.~Kebarighotbi and C.~G. Cassandras, ``Revisiting the optimality of the
  c$\mu$-rule with stochastic flow models,'' in \emph{Decision and Control,
  2009 held jointly with the 2009 28th Chinese Control Conference. CDC/CCC
  2009. Proceedings of the 48th IEEE Conference on}.\hskip 1em plus 0.5em minus
  0.4em\relax IEEE, 2009, pp. 2304--2309.

\bibitem{down2011dynamic}
D.~G. Down, G.~Koole, and M.~E. Lewis, ``Dynamic control of a single-server
  system with abandonments,'' \emph{Queueing Systems}, vol.~67, no.~1, pp.
  63--90, 2011.
\end{thebibliography}

\end{document}